\documentclass[sigconf]{acmart}
\renewcommand\footnotetextcopyrightpermission[1]{} 
\usepackage{booktabs}
\usepackage{amsmath,amsfonts}
\usepackage[ruled,vlined]{algorithm2e}
\usepackage{algpseudocode}

\usepackage{graphicx}
\usepackage{textcomp}
\usepackage{xcolor}

\usepackage{amsthm}
\usepackage{amsmath}
\theoremstyle{definition}
\newtheorem{definition}{Definition}

\usepackage{multirow}
\usepackage{subcaption}
\usepackage{hyperref}
\usepackage{float}

\usepackage[utf8]{inputenc}
\usepackage{enumitem}
\usepackage{tikz}

\makeatletter
    \newlist{treelist}{itemize}{10}
    \setlist[treelist]{label=\treelist@label}

    \tikzset{treelist line/.style={thick, line cap=round, rounded corners}}
    \def\treelist@label{%
        \begin{tikzpicture}[remember picture, baseline={([yshift=-.6ex] treelist-bullet-\the\enit@depth.center)}]
            \draw [treelist line] (0, 0) -- node (treelist-bullet-\the\enit@depth) {} ++(.5em, 0);
        \end{tikzpicture}%
        \ifnum\enit@depth>1
            \tikz[remember picture, overlay] \draw [treelist line] (treelist-bullet-\the\numexpr\enit@depth-1\relax.center) |- (treelist-bullet-\the\enit@depth.center);%
        \fi
    }
\makeatother

\setcopyright{none}
\makeatletter
\renewcommand\@formatdoi[1]{\ignorespaces}
\makeatother
\renewcommand{\footnotetextcopyrightpermission}[1]{} 
\begin{document}
\title{Effective Community Search on Large Attributed Bipartite Graphs}

\author{Zongyu Xu}
\affiliation{
  \institution{Nanjing University of Science and Technology}
}
\email{zongyu.xu@njust.edu.cn}

\author{Yihao Zhang}
\affiliation{
  \institution{Nanjing University of Science and Technology}
}
\email{yhzhangeg@163.com}

\author{Long Yuan}
\authornote{*Corresponding authors}
\affiliation{
  \institution{Nanjing University of Science and Technology}
}
\email{longyuan@njust.edu.cn}

\author{Yuwen Qian}
\affiliation{
  \institution{Nanjing University of Science and Technology}
}
\email{admon@njust.edu.cn}

\author{Zi Chen}
\affiliation{
  \institution{East China Normal University}
}
\email{zchen@sei.ecnu.edu.cn}

\author{Mingliang Zhou}
\authornote{*Corresponding authors}
\affiliation{
  \institution{Chongqing University}
}
\email{mingliangzhou@cqu.edu.cn}

\author{Qin Mao}
\affiliation{
  \institution{Qiannan Normal Coll Nationalities}
}
\email{345379197@qq.com}

\author{Weibin Pan}
\affiliation{
  \institution{North Information Control Research Academy Group Co.}
}
\email{nnupwb@163.com}

\begin{abstract}
Community search over bipartite graphs has attracted significant interest recently. In many applications such as user-item bipartite graph in E-commerce, customer-movie bipartite graph in movie rating website, nodes tend to have attributes, while previous community search algorithm on bipartite graphs ignore attributes, which makes the returned results with poor cohesion with respect to their node attributes. In this paper, we study the  community search problem on attributed bipartite graphs. Given a query vertex q, we aim to find attributed $\left(\alpha,\beta\right)$-communities of $G$, where the structure cohesiveness of the community is described by an $\left(\alpha,\beta\right)$-core model, and the attribute similarity of two groups of nodes in the subgraph is maximized. In order to retrieve attributed communities from bipartite graphs, we first propose a basic algorithm composed of two steps: the generation and verification of candidate keyword sets, and then two improved query algorithms Inc and Dec are proposed. Inc is proposed considering the anti-monotonity property of attributed bipartite graphs, then we adopt different generating method and verifying order of candidate keyword sets and propose the Dec algorithm. After evaluating our solutions on eight large graphs, the experimental results demonstrate that our methods are effective and efficient in querying the attributed communities on bipartite graphs.
\end{abstract}

\maketitle

\section{Introduction}

With the proliferation of graph data, research efforts have been devoted to many fundamental problems in managing and analyzing graph data \cite{yuan2016diversified,yuan2017effective,chen2018exploring,chen2021efficient,huang2017community,yang2018multi,wu2019towards,chen2020efficient,DBLP:journals/pvldb/ZhangYLQZ21,DBLP:journals/pvldb/HaoYZ21,YangZWLXJ21,YangFZLJ21,yang2021corporate}.  Bipartite graphs are widely used to represent the relationships between two different types of entities in many real-world applications, such as user-page networks \cite{DBLP:conf/www/BeutelXGPF13,DBLP:journals/ijprai/QiaoFHZ21}, customer-product networks \cite{DBLP:conf/sigir/WangVR06,DBLP:journals/ijprai/QiHZY22}, collaboration networks \cite{DBLP:conf/spire/Ley02,DBLP:journals/ijprai/CaiMGW22}, gene co-expression networks \cite{DBLP:journals/isci/KaytoueKND11,DBLP:journals/ijprai/ZhuGZF22}. In these practical networks, community structure naturally exists, and a number of cohesive subgraph models (e.g., $\left(\alpha,\beta\right)$-core \cite{DBLP:conf/www/LiuYLQZZ19}, bitruss \cite{DBLP:conf/icde/Wang0Q0020}, and biclique \cite{DBLP:journals/pvldb/LyuQLZQZ20}) are proposed to capture the communities in the bipartite graphs. Following these models, community search over bipartite graphs that aims to find densely connected subgraphs satisfying specified structural cohesiveness conditions has been studied in  applications such as anomaly detection \cite{DBLP:journals/pvldb/LyuQLZQZ20}, personalized recommendation \cite{DBLP:journals/cn/KumarRRT99}, and  gene expression analysis \cite{DBLP:journals/tcbb/MadeiraO04}.

In the aforementioned real-world applications, the entities modeled by the vertices of bipartite graphs often have properties represented by text strings or keywords. When performing community search over such bipartite graphs, previous studies often only focus on the structural cohesiveness of communities but ignore the attributes of the vertices. However, these attributes are important for making sense of communities\cite{2020A,BERAHMAND2021104933,2022A}, and taking the attributes into consideration provides more personalization and interpretation regarding the returned results\cite{DBLP:journals/pvldb/HuangL17,DBLP:journals/pvldb/FangCLH16}, while there are few researches on community search based on attributed bipartite graphs.

\begin{figure}[H]\label{fig}
\centering
\includegraphics[width=85mm]{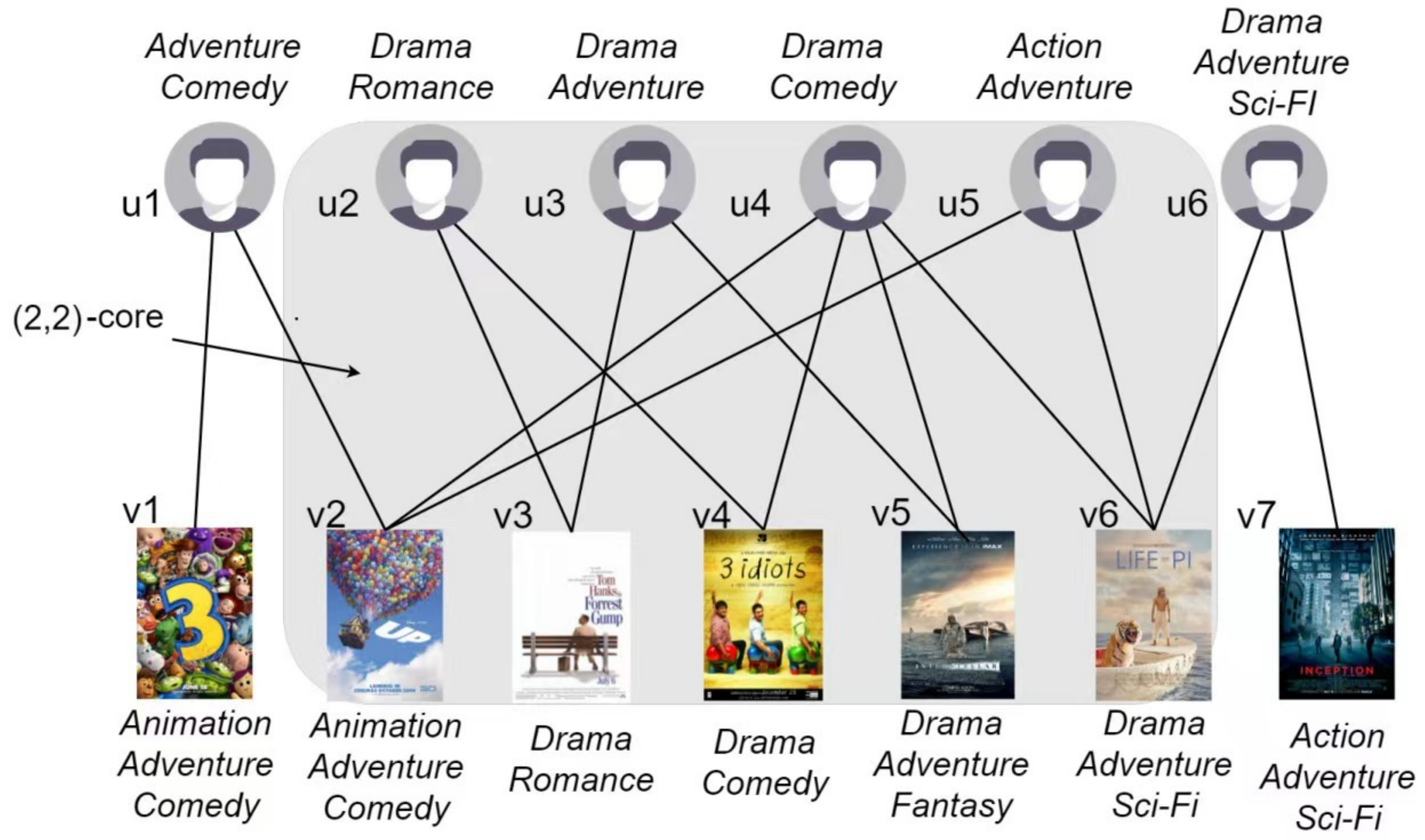}
\vspace{-0.2cm}
\caption{A customer-movie network}
\vspace{-0.5cm}
\end{figure}

Motivated by this, we study the $ attributed\ (\alpha,\beta) $-$community$  $ search $ problem on attributed bipartite graphs in this paper. Specifically, given an attributed bipartite graph $G$ and a query vertex $q \in G$,  we aim to find one or more attributed communities in $G$ such that these communities meet both structure cohesiveness (e.g., each vertex in upper layer has at least $\alpha$ neighbors and each vertex in lower layer has at least $\beta$ neighbors) and keyword cohesiveness (e.g., vertices in the same layer share the most keywords).

\noindent{\bf{Applications.}} Attributed $ (\alpha,\beta)$-community has many real-world applications. For example,

\begin{itemize}[leftmargin=*]
  \item Personalized product recommendation. Attributed $(\alpha,\beta)$-commu-nity can be used to recommend personalized products. Consider the sub customer-movie subnetwork of IMDB (\url{https://www.imdb.com}), where the vertices in the upper layer represent the consumers and the associated attributes describe his or her preference for movies, the vertices in the lower layer represent the movies and the associated attributes describe its genres. The platforms can utilize the attributed $(\alpha,\beta)$-community model to provide personalized recommendation. For example, as Fig.1 shows, if we regard $u_2$ as the query customer, we can find a (2,2)-community composed of viewers $\{u2, u3, u4, u5\}$ and movies $\{v2, v3, v4, v5, v6\}$. In this community “u2” who prefer “Drama” and “Romance” movies may not be interested in “v2”. We further consider the keyword cohesiveness of this community and find an attributed (2,2)-community containing viewers $\{u2, u3, u4\}$ who share the same preference for “Drama” movies and the movies $\{v3, v4, v5\}$ with genre “Drama”. We can recommend the movie “v5” which the user is likely to be interested in to the query viewer “u2”.
  \item Team Formation. In a bipartite graph composed of developers and projects, an edge between a developer and a project indicates that the developer participates in the project, the keywords of developers show their skills while that of projects indicate the technology it requires. When there is a new project to complete, a developer may wish to form a team as cohesive as possible with all developers in this team having the skills that the project requires, which can be supported by an attributed $ (\alpha,\beta)$-community search over the bipartite graph through specifying keywords of the new project.
\end{itemize}

Although attributed $ (\alpha,\beta)$-community search is useful in real applications. it is still inapplicable if the search cannot be finished efficiently, considering that attributed bipartite graph can be very large, and the (structure and keyword) cohesiveness criteria can be complex to handle. A simple way is first to consider all the possible attribute combinations, and then return the corresponding  $(\alpha, \beta)$-community that have the most shared attributes. However, the possible number of attribute combinations is exponential, which makes this approach infeasible in practice.

To address this problem, we observe that the attributed $ (\alpha,\beta)$-community owns the anti-monotonicity property, namely, for a given set $\mathcal{A}$ of attributes, if it appears in every vertex of an attributed $ (\alpha,\beta)$-community, then every subset $\mathcal{A}'$ of $\mathcal{A}$, there exists an attributed $ (\alpha,\beta)$-community in which every vertex contains   $\mathcal{A}'$. Following this observation, we devise efficient algorithms which can significantly reduce the search space when compute the results.

\noindent{\bf{Contributions.}} In this paper, we make the following contributions.

\begin{itemize}[leftmargin=*]
  \item The first work on attributed $(\alpha,\beta)$-community search over attributed bipartite graphs. In this paper, we propose the $(\alpha,\beta)$-community search problem. To the best of our knowledge, this is the first work on attributed $ (\alpha,\beta)$-community search. 
  \item Efficient algorithms to conduct the $(\alpha,\beta)$-community search. Based on the anti-monotonicity property, we devise efficient algorithms to conduct the $(\alpha,\beta)$-community search.
  \item Extensive experiments on real datasets. We conduct extensive experiments to evaluate the performance of the proposed algorithms. The experimental results demonstrates the efficiency of our proposed algorithms.
\end{itemize}

\noindent{\bf{Outline.}} The remainder of this paper is organized as follows. Section 2 presents some related works. Section 3 describes the proposed problem and definitions. A basic solution, enumerating all possible keyword sets and searching for $(\alpha,\beta)$-communities with the most shared keywords, is described in Section 4. Section 5 describes two more efficient algorithms generating and verifying candidate keyword sets in different ways. Section 6 discusses the obtained results with our approaches. Finally, conclusion will be found in Section 7.

\section{Related Work}
\subsection{Community search on unipartite graphs.} Community search performed on unipartite graphs usually using different cohesiveness models such as k-core\cite{seidman1983network}, k-truss \cite{cohen2008trusses}, clique\cite{fang2019efficient}. For a detailed survey, see Ref.~\cite{fang2020survey}. Based on k-core, two online algorithms and one index-based algorithm for k-core community search on unipartite graphs are studied, Cui et al.\cite{cui2014local} propose a local search algorithm, Sozio et al.\cite{sozio2010community}propose a global search algorithm, Barbieri et al.\cite{barbieri2015efficient} propose a tree-like index structure, and Wu et al.\cite{2021Efficient} study the maximal personalized influential community search. Using k-core, Fang et al.\cite{DBLP:journals/pvldb/FangCLH16,2017Effective,2019Effective} further integrate the attributes of vertices to identify community and then the spatial locations of vertices are also considered to identify community\cite{fang2017effective, wang2018efficient,DBLP:journals/jsa/JiWZCLZ21}. For the truss-based community search, Huang et al.\cite{huang2014querying} propose the triangle-connected k-truss community model and then study the closest model.\cite{huang2015approximate}, Akbas et al. \cite{akbas2017truss} also study the triangle-connected k-truss community model and propose an index-based search algorithm. Acquisti et al.\cite{acquisti2006imagined} present an efficient k-clique component detection algorithm and Yuan et al.\cite{yuan2017index} study the problem of densest clique percolation community search.

\subsection{Community search/detection on bipartite graphs.} On bipartite graphs, several existing works \cite{ding2017efficient,he2021exploring,DBLP:conf/www/LiuYLQZZ19,liu2020efficient} extend the k-core model on unipartite graph to the $(\alpha,\beta)$-core model. Ding et al.\cite{ding2017efficient}extend the linear k-core mining algorithm to compute $(\alpha,\beta)$-core. He et al.\cite{he2021exploring} first consider both tie strength and vertex engagement on bipartite graphs and propose a novel cohesive subgraph model. Liu et al.\cite{DBLP:conf/www/LiuYLQZZ19,liu2020efficient} present an efficient algorithm based on a novel index to compute $(\alpha,\beta)$-core in linear time regarding the result size. Based on the butterfly structure, Sariyuce et al.\cite{sariyuce2018peeling}, Wang et al.\cite{wang2019vertex,DBLP:conf/icde/Wang0Q0020}, Zou et al.\cite{zou2016bitruss} study the bitruss model in bipartite graphs which is the maximal subgraph where each edge is contained in at least k butterflies. Zhang et al.\cite{zhang2014finding} study the biclique enumeration problem. zhang et al.\cite{zhang2021pareto} are the first to consider both structure cohesiveness and weight of vertices on bipartite graphs and then propose a novel cohesive subgraph model. Wang et al.\cite{wang2021efficient} present a novel index structure and study the significant community search problem on weighted bipartite graphs, which is the first to study community search on bipartite graphs. However, community search on attributed bipartite graphs remains largely unexplored.

\section{Problem Definition}
Our problem is defined over an undirected attributed bipartite graph $ G=(U,V,E)$, which consists of nodes divided into two separate sets, $ U $ and $ V $, such that every edge connects one node in $ U $ to another node in $ V $. We use $ U(G) $ and $ V(G)$ to denote the two disjoint node sets of $ G $ and $ E(G) $ to represent the edge set of $ G $. Each vertex $ u\in U(G)\ ( v\in V(G)) $ is associated with a set of keywords denoted by $ W_U(u)\ ( W_V(v)) $. An edge $ e $ between two vertices $ u $ and $ v $ in $ G $ is denoted as $ (u,v) $. We denote the number of nodes in $ U(G) $ and $ V(G)$ as $ n_u $ and $ n_v $, the total number of nodes as $ n $ and the number of edges in $ E(G) $ as $ m $. The set of neighbors of a vertex $ u $ in $ G $ is denoted as $ N(u,G)=\{v\in V(G)\vert(u,v)\in E(G)\} $, and the degree of $ u $ is denoted as $ deg(u,G)=\vert N(u,G)\vert $. Table 1 lists the symbols used in the paper.


\begin{table}[th]
	\caption{Symbols and meanings}
		\label{tab1}
	\begin{center}
      \resizebox{85mm}{25mm}{
		\begin{tabular}{c|l}
			\toprule[1.3pt]
           \textbf{Symbol} & \textbf{Meaning} \\ \hline
  \hline
            G(U,V,E)     & An attributed bipartite graph with vertex set U and V, and  \\
                                              & edge set E         \\
           \hline
           $ W_U(u) $             & The keyword set of vertex u in U(G)         \\     
             \hline                       
           $ W_V(v)$              & The keyword set of vertex v in V(G)         \\  
  \hline                      
           $ deg(u,G)$                    & The degree of vertex u in U(G)          \\  
  \hline  
           $ deg(v,G)$                    & The degree of vertex v in V(G)          \\    
  \hline              
           $G[S_u^{'},S_v^{'}]$                & The largest connected subgraph of G s.t. q$ \in G[S_u^{'},S_v^{'}] $, and \\
                                                & $ \forall u\in G[S_u^{'},S_v^{'}], S_u^{'} \subseteq W_U(u) $, $ \forall v\in G[S_u^{'},S_v^{'}],S_v^{'} \subseteq W_V(v) $   \\     
  \hline                     
           $G_{(\alpha,\beta)}[S_u^{'},S_v^{'}]$       & The largest connected subgraph of G s.t. q$ \in G_{(\alpha,\beta)}[S_u^{'},S_v^{'}] $, \\
                                                & and $ \forall u\in G_{(\alpha,\beta)}[S_u^{'},S_v^{'}], deg(u,G)\geq\alpha, S_u^{'} \subseteq W_U(u) $, $ \forall v\in $  \\
                                                     & $  G_{(\alpha,\beta)}[S_u^{'},S_v^{'}], deg(v,G)\geq\beta, S_v^{'} \subseteq W_V(v) $   \\                                    
			\bottomrule[1.2pt]
		\end{tabular}}
	\end{center}
\end{table}

\begin{definition}[$ (\alpha,\beta)$-Core]\label{def1}
Given a bipartite graph $ G $ and two positive integers $ \alpha $ and $ \beta $, a subgraph $ C_{\alpha,\beta} $ is an $ (\alpha,\beta) $-core of $ G $ if $ deg(u,C_{\alpha,\beta})\geq\alpha $ for each $ u\in U(C_{\alpha,\beta}) $ and $ deg(v,C_{\alpha,\beta})\geq\beta $ for each $ v\in V(C_{\alpha,\beta}) $.
\end{definition}

\begin{figure}\label{fig}
\centering
\includegraphics[width=85mm]{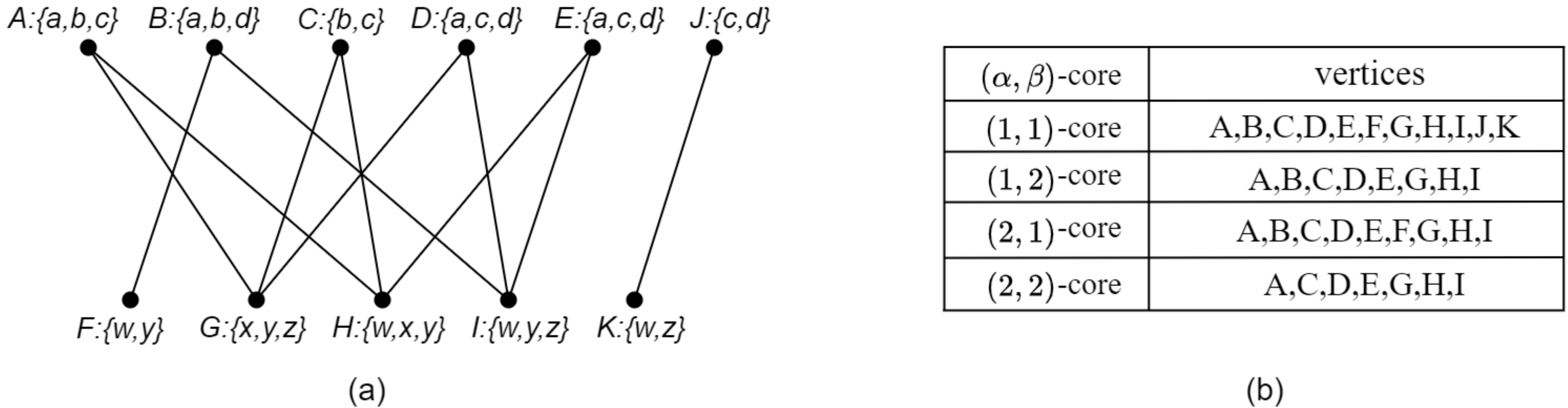}
\vspace{-0.3cm}
\caption{Illustrating the $ (\alpha,\beta) $-core}
\vspace{-0.5cm}
\end{figure}

\begin{example} In Fig.2(a), $\{A, C, D, E, G, H, I\}$ is a (2,2)-core. The (1,1)-core has vertices $\{A, B, C, D, E, F, G, H, I, J, K\}$, and is composed of two (1,1)-core components: $\{A, B, C, D, E, F, G, H, I\}$ and $\{J, K\}$. Each $(\alpha,\beta)$-core in Fig.2(a) is listed in Fig.2(b).	
\end{example}

\begin{definition}[$ (\alpha,\beta) $-Connected Component]\label{def2}
Given a bipartite graph $ G $ and its $ (\alpha,\beta) $-core, $ C_{\alpha,\beta} $, a subgraph $ G_{\alpha,\beta} $ is an $ (\alpha,\beta)$-connected component if (1)$ G_{\alpha,\beta}\subseteq C_{\alpha,\beta} $ and $ G_{\alpha,\beta} $ is connected; (2)$ G_{\alpha,\beta} $ is maximal.
\end{definition}

\begin{definition}[$ (\alpha,\beta) $-Community]\label{def3}
Given a vertex $ q $, we call the $ (\alpha,\beta) $-connected component containing $ q $ the $ (\alpha,\beta) $-community, denoted as $ G_{\alpha,\beta}(q) $.
\end{definition}

\begin{definition}[Attributed $ (\alpha,\beta) $-Community] \label{def4}
Given an attributed bipartite graph $ G $, two positive integers $ \alpha $ and $ \beta $, a query vertex $ q $ and a keyword set $ S\subseteq W(q) $ (i.e., $ q\in U(G)$), a subgraph $ g $ is an attributed $ (\alpha,\beta) $-community of $ G $ if it satisfies the following constraints:
\begin{enumerate}
 \item  \textbf{Connectivity Constraint}. $ g $ is a connected subgraph which contains $ q $.

 \item  \textbf{Structure Cohesiveness Constraint}. $\forall u\in U(g)$, $ deg(u,g)$ $\geq$$\alpha $ and $\forall v\in V(g) $, $ deg(v,g)\geq\beta $.

 \item  \textbf{Keyword Cohesiveness Constraint}. The size of $(\vert L_U(g)\vert+\vert L_V(g)\vert)$ is maximal, where $ L_U(g)=\cap_{u\in U(g)}(W_U(u)\cap S) $ represents the set of keywords shared in $ S $ by all vertices of $ U(g) $ and $ L_V(g)=\cap_{v\in V(g)}(W_V(v)) $ represents the set of keywords shared by all vertices of $ V(g) $.

 \item  \textbf{Maximality Constraint.} There exists no other $ g^{'}\supset g $ satisfying above constraints with $ L_U(g^{'})=L_U(g) $ and $ L_V(g^{'})=L_V(g) $.

\end{enumerate}
\end{definition} 

\begin{figure}[H]\label{fig}
\centering
\includegraphics[width=80mm]{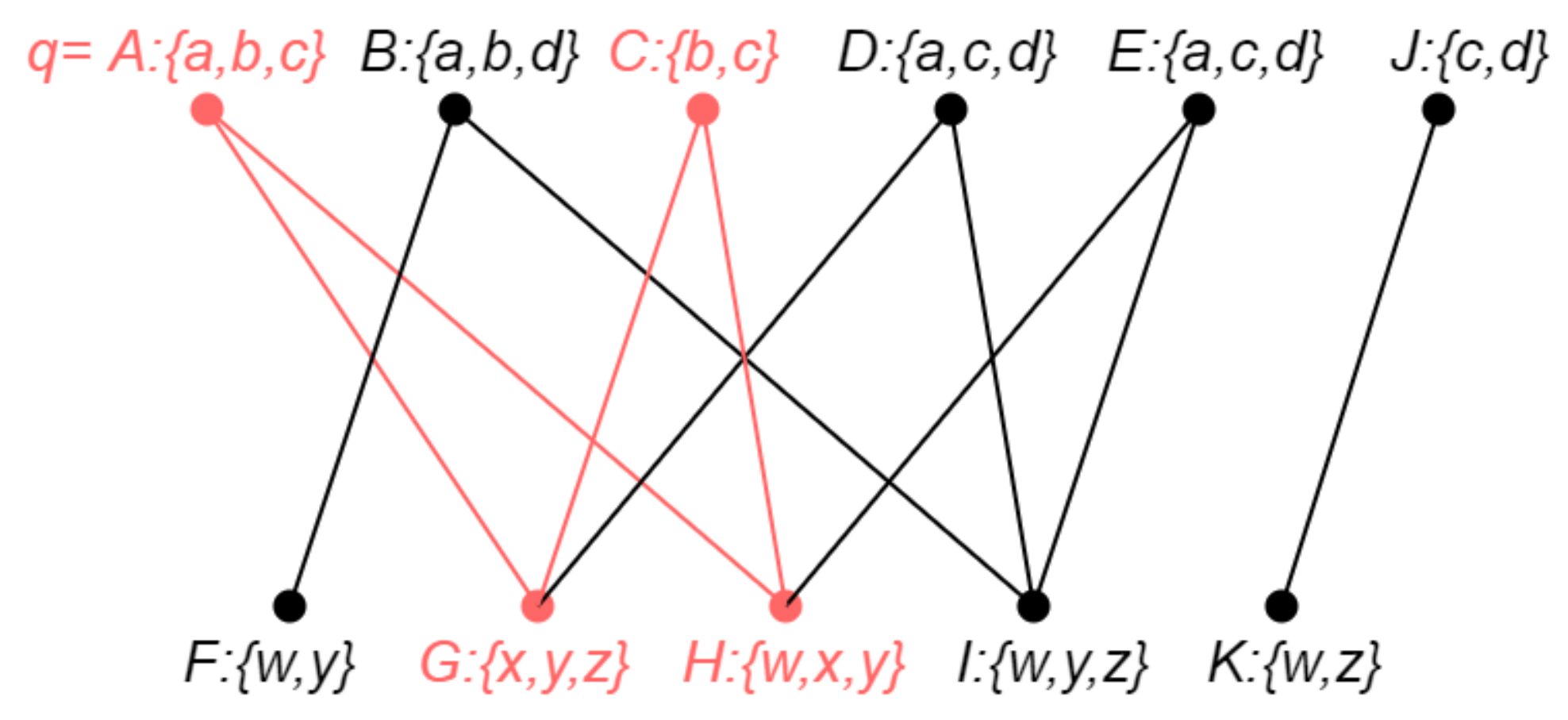}
\vspace{-0.3cm}
\caption{Illustrating an attributed $ (2,2) $-community $ g $}
\vspace{-0.5cm}
\end{figure}

\begin{example}
	Considering the bipartite graph G in Fig.2(a), let q=A, $\alpha $=2, $ \beta$=2. If $ S $=$\{a,b,c\}$, we can find an attributed $ (2,2) $-commu-nity $ g $ as Fig.3 illustrates (in red corlor), whose shared keyword set $ L_U(g)=\{ b,c\} $, $ L_V(g)=\{ x,y\}$.
\end{example}

\noindent\textbf{Problem Statement.} Given an attributed bipartite graph $ G $, parameters $ \alpha $ and $ \beta $, a query vertex $ q $ and a keyword set $ S\subseteq W(q) $, the $ attributed\ (\alpha,\beta) $-$community\ search $ problem aims to find the attributed $ (\alpha,\beta)$-communities in $ G $. For ease of representation, we regard $ q $ as a vertex in $ U(G) $ in this paper. Since the final result must contains $q$, we regard $ S $ as $S_U$, the maximum keyword set which is possible to be shared by all vertices in $ U(G)$.
\section{Basic Solution}

We use $ G\lbrack S_u,S_v\rbrack $ to denote the largest connected subgraph of $ G $, where each vertex in $ U(G\lbrack S_u,S_v\rbrack)(V(G\lbrack S_u,S_v\rbrack)) $ contains $ S_u(S_v) $ and $ q\in G\lbrack S_u,S_v\rbrack $. We use $ G_{\alpha,\beta}\lbrack S_u,S_v\rbrack $ to denote the largest connected subgraph of $ G\lbrack S_u,S_v\rbrack $, in which every vertex in $ U(G_{\alpha,\beta}\lbrack S_u,S_v\rbrack ) $ has degree being at least $ \alpha $ and every vertex in $ V(G_{\alpha,\beta}\lbrack S_u,S_v\rbrack ) $ has degree being at least $ \beta $. We call $ \{S_u,S_v\} $ a qualified keyword set for the query vertex $ q $ on the graph $ G $, if $ G_{\alpha,\beta}\lbrack S_u,S_v\rbrack $ exists. 

Given a query vertex $q$, a straightforward method to find the attributed $ (\alpha,\beta) $-communities in $ G $ performs three steps. First, for one layer of the bipartite graph which contains q, here we consider it as $ U(G) $ and consider $ S $ as $ S_U $, all nonempty subsets of $ S_U $, $ S_{U1} $, $ S_{U2} $, $...$, $ S_{U(2^{l}-1)}\ (l = \vert S_U\vert) $, are enumerated, and for each $ v\in V(G) $, we put all different keywords in $ W_V(v) $ into $ S_V $ and enumerate all nonempty subsets of $ S_V$$(i.e., S_{V1} $, $ S_{V2} $, $...$, $ S_{V(2^{k}-1)}\ (k = \vert S_V\vert)) $. Then for each set $ \{S_{Ui},S_{Vj}\}(1\leq i\leq2^l-1,1\leq j\leq2^k-1) $, we verify the existence of $ G_{(\alpha,\beta)}\lbrack S_{Ui},S_{Vj}\rbrack $ and compute it when it exists. Finally, we output the subgraphs having the most shared keywords among all $ G_{(\alpha,\beta)}\lbrack S_{Ui},S_{Vj}\rbrack $.

\begin{figure}[H]\label{fig}
\centering
\includegraphics[width=80mm]{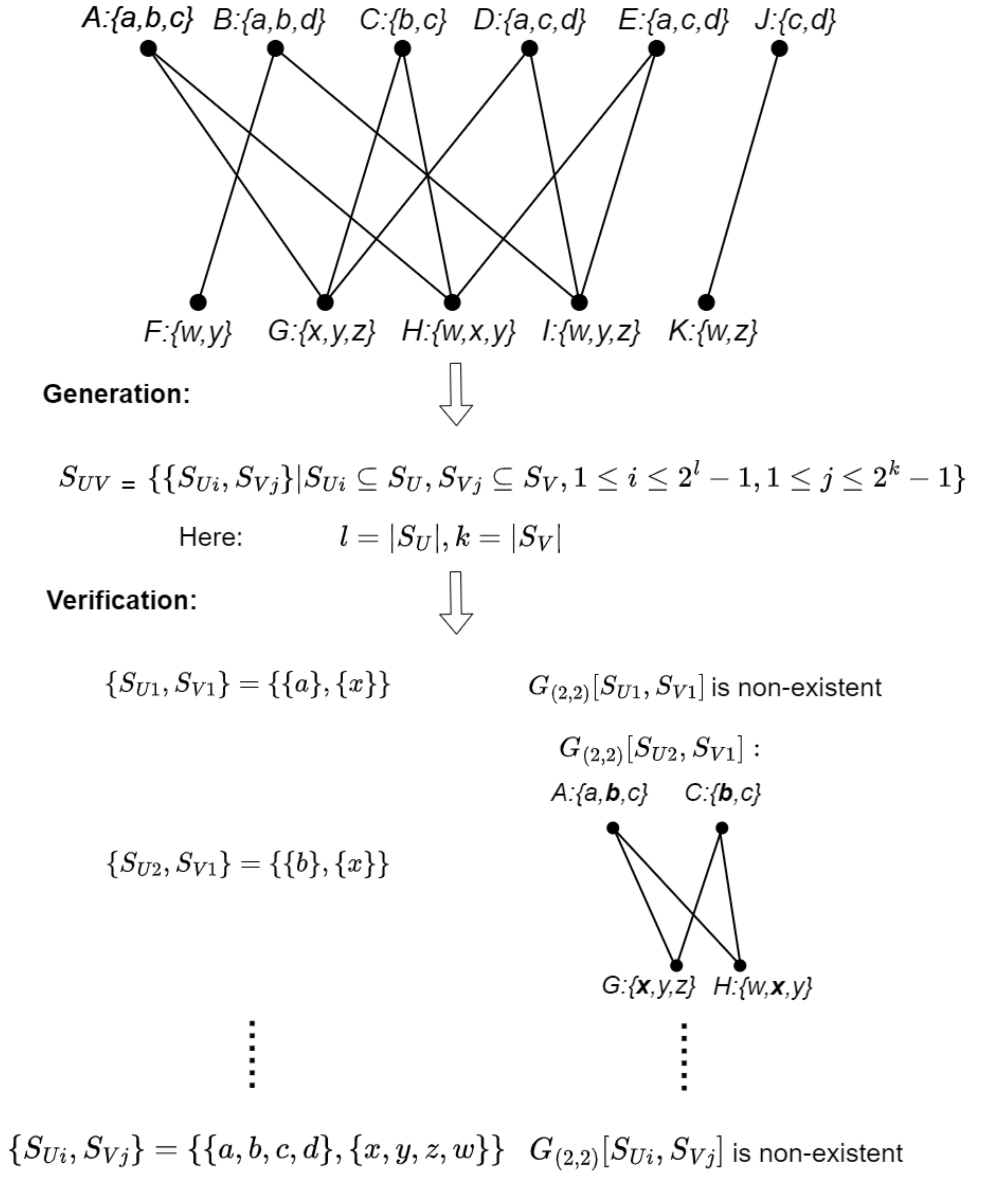}
\vspace{-0.2cm}
\caption{Generation and verification of candidate keyword sets}
\vspace{-0.5cm}
\end{figure}

We can summarize the straightforward method into a two-step framework, generation and verification of candidate keyword sets. Considering the bipartite graph $G$ in Fig.2(a), let $q$=$A$, $\alpha $=2, $ \beta$=2, $ S $=$\{a,b,c\}$, Fig.4 shows how we find attributed $ (2,2) $-communities through the two-step framework, and the the computational complexity for the proposed framework is the same as that for the $Basic$ algorithm mentioned below.

Here we first give the procedure to verify the existence of $G_{\alpha,\beta}$ $(q,G^{'}) $ in a given subgraph $G^{'}$ of $G$ for each given candidate keyword set.

\begin{algorithm}
\caption{Compute $ G_{\alpha,\beta}(q,G^{'}) $}\label{algorithm}
      \For{$ u\in U(G^{'}) $}{
          \If {$ deg(u,G^{'})<\alpha$}{
               remove $ u $ and its incident edges from $ G^{'} $\;
          }
      }
     \While{$ q\in G^{'} $}{
           $ x\leftarrow min_{v\in V(G')}deg(v,G') $\;
           \If{$ x\geq\beta $}{
            return $ G^{'} $\;
           \Else{
                 \For{$ v\in V(G^{'}) $}{
                    \If{$ deg(v,G^{'})<\beta$}{
                        \For{$ p\in N(v) $}{
                             remove $ (p,v) $\;
                             \If{$ deg(p,G^{'})<\alpha$}{
                                 remove $ p $ and its incident edges from $ G^{'}$\;
                              }
                         }remove $v$\;
                     }
                 }
            }
            }
       }
      return $ G_{\alpha,\beta}(q,G^{'}) $\;
\end{algorithm}

\begin{theorem}\label{thm1}
Given a bipartite graph $G$, It takes $O(d_{umax}\cdot (n_u+n_v\cdot d_{vmax}))$ to compute $ G_{\alpha,\beta}(q,G^{'}) $.
\end{theorem}

\begin{proof}
There are $n_u$ nodes in $U(G^{'})$, $n_v$ nodes in $V(G^{'})$ and we denote the largest degree of these nodes in $U(G^{'})(V(G^{'}))$ as $d_{umax}(d_{vmax})$. Removing all $u\in U(G)$ with degree less than $\alpha$ cost $O(n_u\cdot d_{umax})$, and the while loop in line 4-15 cost $O(n_v\cdot d_{vmax}\cdot d_{umax})$.
\end{proof}

\begin{algorithm}
\caption{Basic}\label{algorithm}
      Initialize $ \psi $ using $ S $, $ \varphi $ using $ V(G) $\;
      \While{true}{
      $ max\leftarrow0 $, $ m\leftarrow0 $, $ \phi_m\leftarrow\varnothing $\;
      \For{$ \psi^{'}\in\psi $}{
          \For{$ \varphi^{'}\in\varphi $}{
            find $ G\lbrack\psi^{'},\varphi^{'}\rbrack $ from $ G $\;
            Compute $ G_{(\alpha,\beta)}\lbrack\psi^{'},\varphi^{'}\rbrack $ from $ G\lbrack\psi^{'},\varphi^{'}\rbrack $\;
            \If{$ G_{(\alpha,\beta)}\lbrack\psi^{'},\varphi^{'}\rbrack $ exists}{
            $ m\leftarrow(\vert\psi^{'}\vert+\vert\varphi^{'}\vert) $\;
            \If{$ max\leq m $}{
            $ max\leftarrow m $\;
            $ \phi_m.add(\psi^{'}+\varphi^{'}) $\;
}
}
}
}
}
     \If{$ \phi_m\neq\varnothing $}{
      output the communities of keyword sets in $ \phi_m $\;
}
\end{algorithm}

Based on the straightforward method, we present Algorithm2, a baseline query algorithm called $ Basic $. The input of $ basic $ is a bipartite graph $ G $, a query vertex q, two positive integers $ \alpha $ and $ \beta $, and a set $ S $. It first initializes a set, $ \psi $, of candidate keyword sets with each being a nonempty subset of $ S (i.e.,S_1,S_2,S_3,...,S_{U(2^{l}-1)}(l=\vert S\vert)) $ (line 1). After that, for each vertex in $ V(G) $, we enumerate all nonempty subsets of $ W_V(v) $, put them into $ \varphi $ and ensure that each element in $ \varphi $ appears only once. In the while loop (lines 2–12), it first set $ m = 0 $,indicating the size of current keyword sets, $ max =0 $, indicating the maximal size of all keyword sets and an empty set $ \phi_m $ (line 3) for collecting all the qualified keyword sets. Then for each $ \psi^{'}\in\psi $ and for each $ \varphi^{'}\in\varphi $, it finds $ G_{\alpha,\beta}\lbrack \psi^{'},\varphi^{'}\rbrack $ from $ G_{\alpha,\beta} $ by considering the keyword and degree constraints (line 4-7). If $ G_{\alpha,\beta}\lbrack \psi^{'},\varphi^{'}\rbrack $ exists, the sum of numbers of elements in $ \psi^{'} $ and $ \varphi^{'} $ is recorded by $ m $. Then we compare $ m $ with $ max $. If $ max\leq m $, it then assign $ m $ to $ max $ and put the set of current keywords in $ \psi^{'} $ and $ \varphi^{'} $ into $ \phi_m $ (line 10-12). After checking all the candidate keyword sets in $ \psi $ and $ \varphi $, if there are at least one qualified keyword sets in $ \phi_m $, it output the communities of keyword sets in $ \phi_m $ (line 13-14).

\begin{theorem}\label{thm2}
Given a bipartite graph G, $Basic$ computes $ G_{\alpha,\beta}\lbrack S_u,$ $\ S_v\rbrack $ in $O(n_v\cdot2^{\vert S_v\vert_{max}}\log(n_v\cdot2^{\vert S_v\vert_{max}})+2^{\vert S\vert}\cdot 2^{\vert S_v\vert_{max}}\cdot O(compute$ $\ G_{\alpha,\beta}$$(q,G^{'})))$.
\end{theorem}

\begin{proof}
We use $\vert S_v\vert_{max}$ to represent the $ W_V(v) $ of largest size among all $v\in V(G)$, Initializing $ \psi $ and $ \varphi $ can be completed in $O(2^{\vert S\vert}+n_v\cdot2^{\vert S_v\vert_{max}}\log(n_v\cdot2^{\vert S_v\vert_{max}}))$ and the while loop in line 2-12 costs $O(2^{\vert S\vert}\cdot 2^{\vert S_v\vert_{max}}\cdot O(compute\  G_{\alpha,\beta}(q,G^{'})))$.
\end{proof}

One major drawback of the straightforward method is that we need to compute $ (2^l-1)\times(2^k-1) $ subsets of attributes and verify the existence of corresponding subgraphs (i.e.,$ G_{(\alpha,\beta)}\lbrack S_{Ui},S_{Vj}\rbrack $). For large values of $ l $ and $ k $, the computation overhead makes this method impractical. To alleviate this problem, we study methods to simplify the generation and verification of candidate keyword sets, and propose two improved algorithms.

\section{Improved Attributed $(\alpha, \beta)$-community Search Algorithm}
In this section, we shrink the range of possible candidate keyword sets and develop two more efficient algorithms: the incremental algorithm ($ Inc $) verify the candidate sets from smaller to larger ones while the decremental algorithm ($ Dec $) examine larger candidate sets to smaller ones.

\subsection{The Incremental Algorithm}

Attributed bipartite graphs have the anti-monotonicity property regarding the attributed $ (\alpha,\beta)$-community search, which is  shown in the  following lemma:

\begin{lemma}\label{lem1}
Given a graph $ G $, a vertex $ q\in G $, set $ S_u $ and $ S_v $ of keywords, if there exists a subgraph $ G_{\alpha,\beta}\lbrack S_u,S_v \rbrack $ , then there exists a subgraph $ G_{\alpha,\beta}\lbrack S_u^{'},S_v^{'}\rbrack \supseteq G_{\alpha,\beta}\lbrack S_u,S_v \rbrack $ for any subset $ S_u^{'}\subseteq S_u,S_v^{'}\subseteq S_v $. 
\end{lemma}

\begin{proof}
Based on the definition of $ G_{\alpha,\beta}\lbrack S_u,S_v \rbrack $, each vertex in $ U(G_{\alpha,\beta}\lbrack S_u,S_v \rbrack) $ contains $ S_u $ and each vertex in $ V(G_{\alpha,\beta}\lbrack S_u,S_v \rbrack) $ contains $ S_v $. Consider two new keyword sets $ S_u^{'}\subseteq S_u,S_v^{'}\subseteq S_v $, we can easily conclude that each vertex in $ U(G_{\alpha,\beta}\lbrack S_u,S_v \rbrack) $ contains $  S_u^{'} $ and each vertex in $ V(G_{\alpha,\beta}\lbrack S_u,S_v \rbrack) $ contains  $ S_v^{'} $ as well. Also, note that $ q\in G_{\alpha,\beta}\lbrack S_u,S_v \rbrack $. These two properties imply that there exists one subgraph of $ G $, namely $ G_{\alpha,\beta}\lbrack S_u,S_v \rbrack $, with each vertex in $ U(G) $ has degree being at least $ \alpha $ and each vertex in $ V(G) $ has degree being at least $ \beta $, such that it contains $ q $ and every vertex in its upper(lower) layer contains $ S_u^{'}(S_u^{'}) $. It follows that there exists such a subgraph with maximal size (i.e.,$ G_{\alpha,\beta}\lbrack S_u^{'},S_v^{'}\rbrack $).
\end{proof}

\begin{lemma}\label{lem2}
Given two groups of keyword sets $ \{S_{u1},S_{v1}\} $ and $ \{S_{u2},S_{v2}\} $, if $ G_{\alpha,\beta}\lbrack S_{u1},S_{v1}\rbrack $ and $ G_{\alpha,\beta}\lbrack S_{u2},S_{v2}\rbrack $ exist, we have $ G_{\alpha,\beta}$ $\lbrack S_{u1\cup u2},S_{v1\cup v2}\rbrack\subseteq G_{\alpha,\beta}\lbrack S_{u1},S_{v1}\rbrack\cap G_{\alpha,\beta}\lbrack S_{u2},S_{v2}\rbrack $.
\end{lemma}

\begin{proof} Based on Lemma 1, since $ \{S_{u1},S_{v1}\}\subseteq \{S_{u1\cup u2},S_{v1\cup v2}\}$ and $ G_{\alpha,\beta}\lbrack S_{u1},S_{v1}\rbrack $ exsits, we have $ G_{\alpha,\beta}\lbrack S_{u1\cup u2},S_{v1\cup v2}\rbrack\subseteq G_{\alpha,\beta}\lbrack S_{u1},$ $S_{v1}\rbrack$. For the same reason, we have $ G_{\alpha,\beta}\lbrack S_{u1\cup u2},S_{v1\cup v2}\rbrack\subseteq G_{\alpha,\beta}\lbrack S_{u2},$ $S_{v2}\rbrack$. It directly follows the lemma.
\end{proof}

This lemma implies, if $ \{S_u^{'},S_v^{'}\} $ is generated from $ \{S_{u1},S_{v1}\} $ and $ \{S_{u2},S_{v2}\} $, we can find $ G_{\alpha,\beta}\lbrack S_u^{'},S_v^{'}\rbrack $ from $ G_{\alpha,\beta}\lbrack S_{u1},S_{v1}\rbrack\cap G_{\alpha,\beta}\lbrack S_{u2},$ $S_{v2}\rbrack $ directly. Since every vertex in $ G_{\alpha,\beta}\lbrack S_{u1},S_{v1}\rbrack\cap G_{\alpha,\beta}\lbrack S_{u2},S_{v2}\rbrack $ contains both  $ \{S_{u1},S_{v1}\} $ and $ \{S_{u2},S_{v2}\} $, we do not need to consider the keyword constraint again when finding $ G_{\alpha,\beta}\lbrack S_u^{'},S_v^{'}\rbrack $.

In addition, considering the degree constraint of  $ G_{\alpha,\beta}\lbrack S_u^{'},S_v^{'}\rbrack $, there is a key observation that, if $ {S_u^{'},S_v^{'}} $ is a qualified keyword set, then there are at least $ \beta $ vextices in $ U(G_{\alpha,\beta}\lbrack S_u^{'},S_v^{'}\rbrack) $ containing set $ S_u^{'} $ and $ \alpha $ vertices in $ N(q) $ containing set $ S_v^{'} $. This observation implies, we can generate all the candidate keyword sets directly by using the query vertex $q$ and $q^{'} $ neighbors, without touching other vertices.

Based on above lemmas and observation, we introduce the algorithm $ Inc $. Compared with $ Basic $, it  shrinks the initial candidate keyword sets and can always verify the existence of $ G_{(\alpha,\beta)}\lbrack\psi^{'},\varphi^{'}\rbrack $ within a subgraph of G instead of the entire graph $ G $ , and thus the subgraph for such verification shrinks when the candidate set $ {\psi^{'},\varphi^{'}} $ expands. Therefore, a large sum of redundant computation is reduced during the verification process.

\begin{algorithm}
\caption{Inc}\label{algorithm}
      Initialize $ \psi $ using $ S $, $ \varphi $ using $ N(q) $\;
      generate $ P\{P_1,P_2,...,P_i\} $ and $ Q\{Q_1,Q_2,...,Q_j\} $ by $ \psi $ and $ \varphi $\;
      update $ \psi, \varphi, P, Q $\;
      $ c\leftarrow\varnothing, \phi_l\leftarrow\varnothing, l\leftarrow0$\;
      \For{$ \psi_i\in\psi $}{
          \For{$ \varphi_j\in\varphi $}{
               Compute $ G_{(\alpha,\beta)}\lbrack\psi_i,\varphi_j\rbrack $ from the subgraph induced on $ P_i $ and $ Q_j $\;
               \If{$ G_{(\alpha,\beta)}\lbrack\psi_i,\varphi_j\rbrack $ exists}{
                $ c\leftarrow\{\psi_i,\varphi_j\} $\;
                $ \phi_l.add(<c,G_{\alpha,\beta}\lbrack c\rbrack>) $\;
                }
           }
      }
      \While{$ \phi_l\neq\varnothing $}{
            \For{$ <c_1,G_{\alpha,\beta}\lbrack c_1\rbrack>\in \phi_l $}{
                 \For{$ <c_2,G_{\alpha,\beta}\lbrack c_2\rbrack>\in \phi_l $}{
                      $ G\lbrack c_1\cup c_2\rbrack\leftarrow G_{\alpha,\beta}\lbrack c_1\rbrack\cap G_{\alpha,\beta}\lbrack c_2\rbrack $\;
                      Compute $ G_{(\alpha,\beta)}\lbrack c_1\cup c_2\rbrack $ from $ G\lbrack c_1\cup c_2\rbrack $\;
                      \If{$ G_{\alpha,\beta}\lbrack c_1\cup c_2\rbrack $ exists}{
                          $ \phi_{l+1}.add(<c_1\cup c_2,G_{\alpha,\beta}\lbrack c_1\cup c_2\rbrack>) $\;
                          $ l\leftarrow l+1 $\;
                       }
                  }
   
             }
      }
      find $ c $ when $ \vert c\vert $ is maximum from $ \phi_0 $ to $ \phi_{l-1} $\;
      output $ G_{\alpha,\beta}\lbrack c\rbrack $\;
\end{algorithm}

Algorithm 3 presents $Inc$. First it initializes a set, $ \psi\{\psi_1,\psi_2,...,\psi_i\} $, of candidate keyword sets with each being a keyword of $ S $. Then for each $ v\in N(q) $, it puts each keyword in $ W_V(v) $ into $ S_V $ and initializes a set, $ \varphi\{\varphi_1,\varphi_2,...,\varphi_j\} $, of candidate keyword sets with each being a keyword of $ S_V $ (line 1). For each candidate keyword set $ \psi_i(\varphi_j) $ in $ \psi(\varphi) $, it traverse $ G $ and put nodes containing $ \psi_i(\varphi_j) $ into $ P_i(Q_j) $ (line 2). Considering the key observation that, if $ \psi_i(\varphi_j) $ is a qualified keyword set, then there are at least $ \beta $ nodes in $ U(G)$ containing $ \psi_i $ and $ \alpha $ nodes in $ V(G)$ containing $ \varphi_j $, so it removes $ P_i $ and $ \psi_i $ if $ \vert P_i\vert<\beta $, and removes $ Q_j $ and $ \varphi_j $ if $ \vert Q_j\vert<\alpha $ as well (line 3). Then, we set $l = 0 $, indicating the sizes of current keyword sets, and initialize a set $ \phi $ of $ <c , G_{\alpha,\beta}\lbrack c\rbrack> $ pairs. In a $ <c ,G_{\alpha,\beta}\lbrack c\rbrack> $ pair, $c$ contains a set, $\psi^{'}$, of keywords from $ \psi$ and a set, $\varphi^{'}$, of keywords from $\varphi $, and $ G_{\alpha,\beta}\lbrack c\rbrack $ is an $ (\alpha,\beta)$-community of $ G $ where each vertex in $ U(G_{\alpha,\beta}\lbrack c\rbrack) $ contains $\psi^{'}$ and each vertex in $ V(G_{\alpha,\beta}\lbrack c\rbrack) $ contains $\varphi^{'}$ (line 4). $ \forall\psi^{'}\in\psi $ and $ \forall\varphi^{'}\in\varphi $, we verify the existence of $ G_{(\alpha,\beta)}\lbrack\psi^{'},\varphi^{'}\rbrack $ and put the qualified $ <c ,G_{\alpha,\beta}\lbrack c\rbrack> $ pairs into $ \phi_l $ (line 5-10). In the while loop (lines 11–18), for every two $ <c ,G_{\alpha,\beta}\lbrack c\rbrack> $ pairs, denoted as $ <c_1 ,G_{\alpha,\beta}\lbrack c_2\rbrack> $ and $ <c_2 ,G_{\alpha,\beta}\lbrack c_2\rbrack> $ in $ \phi_l $, we find $ G_{(\alpha,\beta)}\lbrack c_1\cup c_2\rbrack $ from $ G\lbrack c_1\cup c_2\rbrack $, the shared subgraph of $ G_{\alpha,\beta}\lbrack c_2\rbrack $ and $ G_{\alpha,\beta}\lbrack c_2\rbrack $ (line 12-15). If $ G_{(\alpha,\beta)}\lbrack c_1\cup c_2\rbrack $ exists, we put the pair of $ c_1\cup c_2 $ and $ G_{(\alpha,\beta)}\lbrack c_1\cup c_2\rbrack $ into the set $ \phi_{l+1} $ (line 16-17). When $\phi_l $ is empty, we stop the loop. Next, we  look for the qualified keyword sets $ c $, which contain the most keywords, from $\phi_0 $ to $\phi_{l-1}$. Finally, we output the communities of keyword sets $ c $.

\begin{theorem}\label{thm3}
Given a bipartite graph G, $Inc$ computes $ G_{\alpha,\beta}\lbrack c\rbrack $ in $O((\vert S\vert+\vert S_V\vert-1)\cdot \vert S\vert\cdot\vert S_V\vert(\vert S\vert\cdot\vert S_V\vert +O(Compute\ G_{\alpha,\beta}(q,G))))$.
\end{theorem}

\begin{proof}
In Algorithm 3, we use $d$ to denote the degree of $q$ and $\vert S_v\vert_{max}$ to represent the $ W_V(v) $ of largest size among all $v\in N(q)$, lines 1 can be completed in $O(\vert S\vert+d\cdot\vert S_v\vert_{max}\log(d\cdot\vert S_v\vert_{max}))$ time. Line 2-3 can be completed in $O(n_u\cdot\vert S\vert+n_v\cdot d\cdot\vert S_v\vert_{max}\log(d\cdot2^{\vert S_v\vert_{max}}))$ time. Line 5-10 can be completed in $O(\vert S\vert\cdot\vert S_V\vert\cdot O(Compute$ $\ G_{\alpha,\beta}(q,G)))$ time. In while loop, each time it takes $O(\vert S\vert\cdot\vert S_V\vert(\vert S\vert\cdot\vert S_V\vert +O(Compute\ G_{\alpha,\beta}(q,G))))$ time to find qualified communities and put them into a new set $ \phi_{l+1}$, in the worst case, it runs $(\vert S\vert+\vert S_V\vert-1)$ times.
\end{proof}

\begin{example} Considering $G$ in Fig.2(a), let $q$=$A$, $\alpha $=2, $ \beta$=2 and $ S $=$\{a,b,c\}$, Fig.5(a) shows a (2,2)-core of $G$. By Algorithm 3, we first find set of keyword sets $ \psi\{\{a\},\{b\},\{c\}\} $, $ \varphi\{\{w\},\{x\},\{y\},\{z\}\} $ and then verify that $ G_{2,2}\lbrack \{b\},\{x\} \rbrack $, $ G_{2,2}\lbrack \{b\},\{y\} \rbrack $, $ G_{2,2}\lbrack \{c\},\{x\} \rbrack $ and $ G_{2,2}\lbrack \{c\},\{y\} \rbrack $ exists as Fig.5(b) and Fig.5(c) show. In the first while loop, we choose 2 qualified keyword sets from $\{\{b,x\},$ $ \{b,y\}, \{c,x\},$ $\{c,y\}\}$ and get their union set (e.t.$ \{bc,xy\}\ from\ \{b,x\}\ and \ \{c,y\} $). By Lemma 2, we only need to verify the new candidate keyword set under nodes in $ G_{2,2}\lbrack \{b\},\{x\} \rbrack $ and $ G_{2,2}\lbrack \{c\},\{y\} \rbrack $. Fig.5(d) shows the final attributed community	$ G_{2,2}\lbrack \{b,c\},\{x,y\}\rbrack $.
\end{example}

\begin{figure}[H]\label{fig}
\centering
\includegraphics[width=85mm]{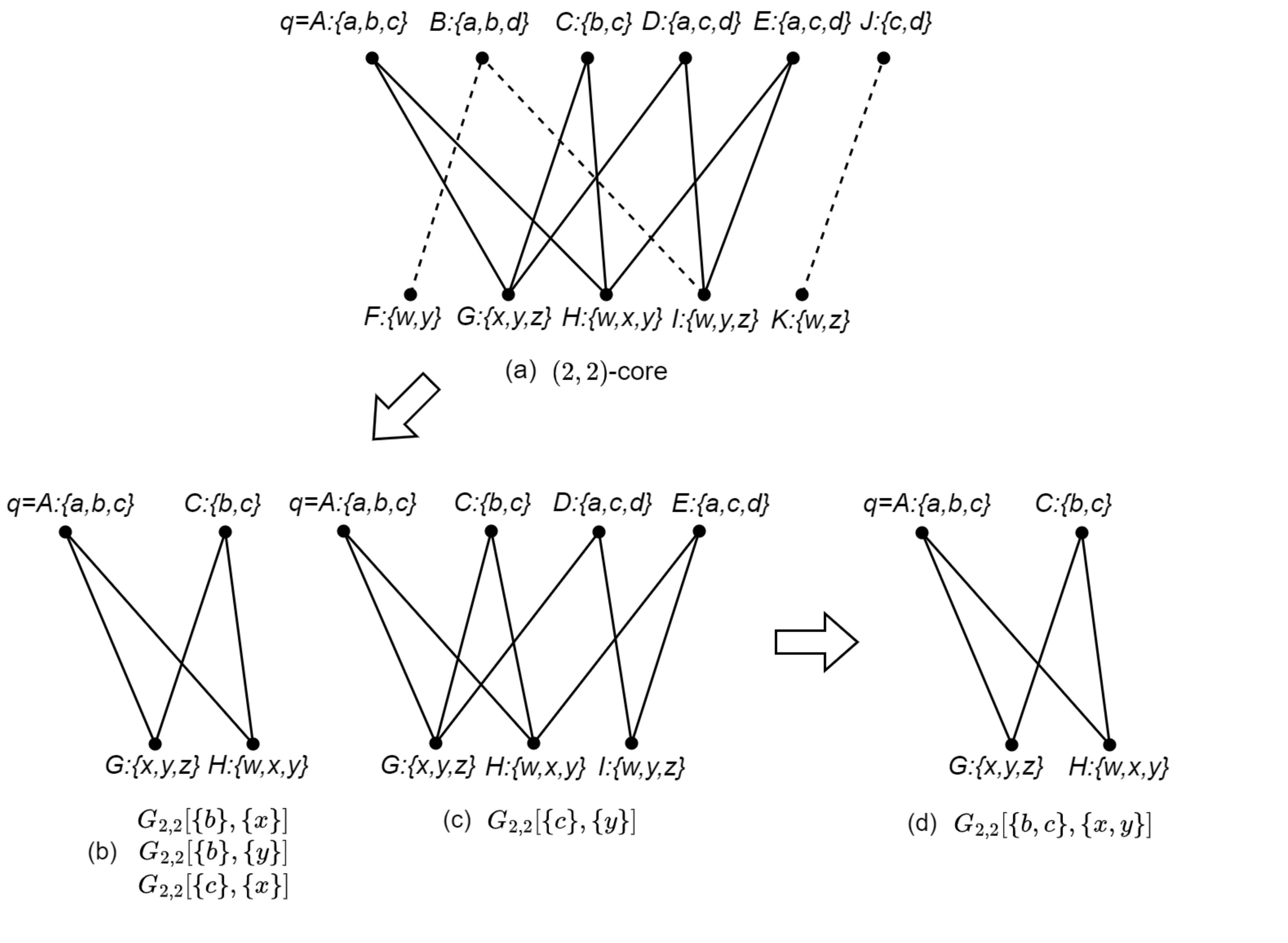}
\vspace{-0.5cm}
\caption{An example of finding $(2,2)$-community by $ Inc $ method}
\vspace{-0.5cm}
\end{figure}

\subsection{The Decremental Algorithm}
The decremental algorithm, denoted by $ Dec $, differs from the incremental algorithm on both the generation and verification of candidate keyword sets.

\subsubsection{Generation of candidate keyword sets}

\begin{lemma}\label{lem3}
Given a vertex set V of $q^{'}$s neighbors, a qualified keyword set $S_u$ and a set $S_V$ containing all nonempty subsets of $W_V(v)$. For each $S_v \in S_V $, if less than $\alpha$ vertices in V containing $S_v$, we have  $ G_{\alpha,\beta}\lbrack S_u,S_v \rbrack$ does't exist.
\end{lemma}

\begin{proof}
Assume that $\{S_u,S_v\}$ is a qualified keyword set, then there are at least $\beta$ vertices in $ U( G_{\alpha,\beta}\lbrack S_u,S_v \rbrack)$ containing $S_u$ and $\alpha$ vertices of  $q^{'}s$ neighbors containing $S_v$. This contradicts the condition that less than $\alpha$ vertices in $V$ contains $S_v$, so lemma 3 is proved.
\end{proof}

We generate the candidate keyword sets, $\psi $, of $ U(G) $ by enumerating all nonempty subsets of $ S_U $. For each vertex $ v\in N(q) $, we enumerate all nonempty subsets of $ W_V(v) $ and put them into a new set $ \varphi $, the elements of which are different from each other. Then we update the candidate keyword sets by removing those contained by less than $\alpha $ of $ q^{'} $ neighbors.

\begin{example} Consider a query vertex Q($ \alpha $ = 3)with 5 neighbors in Fig.6(a), where the selected keywords of each vertex are listed in the curly braces. For each neighbor of Q, all nonempty subsets of its keyword sets are generated, as shown in Fig.6(b). We can easily filter out the subset which occurs equal to or more than three times and form the set $ \varphi\{\{x\},\{y\},\{z\},\{x,y\}\} $.	
\end{example}

\begin{figure}[H]\label{fig}
\centering
\includegraphics[width=80mm]{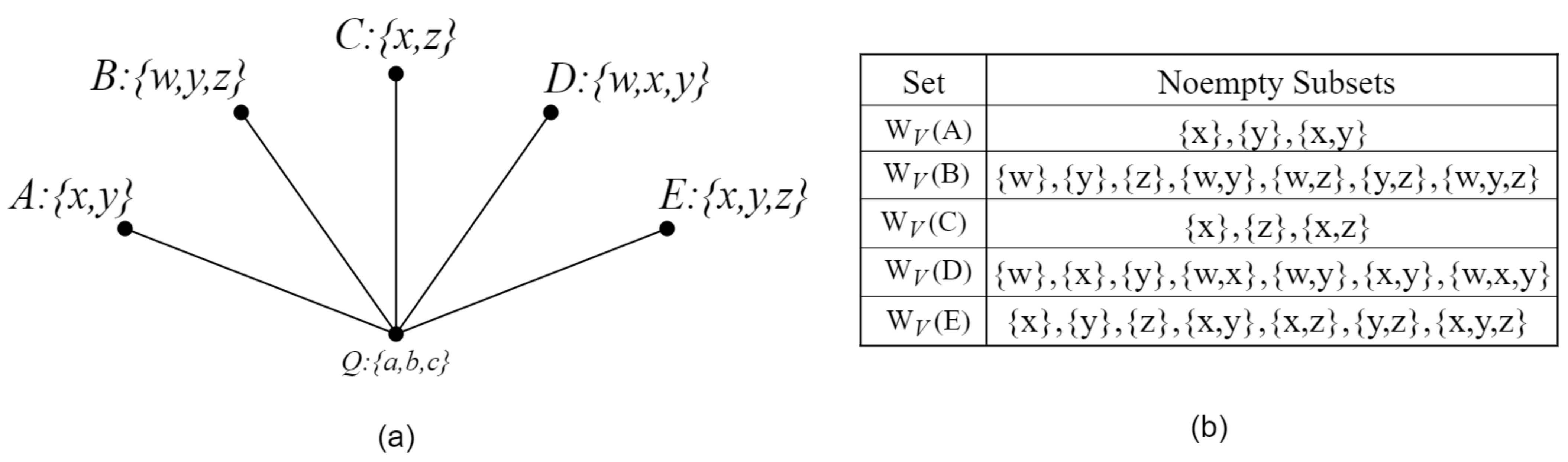}
\vspace{-0.2cm}
\caption{An example of candidate generation in $Dec$ method}
 \vspace{-0.5cm}
\end{figure}

\subsubsection{Verification of candidate keyword sets} 
As candidates can be obtained using $ S $ and $q^{'} $ neighbors directly, we can verify them in a decremental manner (larger candidate keyword sets first and smaller candidate keyword sets later). During the verification process, once finding the attribute $ (\alpha,\beta)$-communities for candidate keyword sets of the same size, $ Dec $ does not need to verify smaller candidate keyword sets. Therefore, compared with the incremental algorithm, $ Dec $ can save the cost of verifying smaller candidate keywords, thus it may be faster practically.

\begin{algorithm}
\caption{Dec}\label{algorithm}
      Initialize $ \psi $ using $ S $, $ \varphi $ using $ N(q) $\;
      create $ P_1,P_2,...,P_i $ and $ Q_1,Q_2,...,Q_j $ by $ \psi $ and $ \varphi $\;
      update $ \psi, \varphi, P, Q $\;
      $ c\leftarrow\varnothing, S\leftarrow\varnothing, ans\leftarrow\varnothing, max\leftarrow0$\;
       \For{$ \psi_i\in\psi $}{
          \For{$ \varphi_j\in\varphi $}{
                $ c\leftarrow\{\psi_i,\varphi_j\} $\;
                $ S.add(c) $\;
           }
      }
      sort $ S\{S_1,S_2,...,S_{i\times j}\} $ in descending order\;
      \For{$ S_k\in S $}{
            \If{$ \vert S_k\vert<max $}{break\;}
            \Else{
            compute $ G_{(\alpha,\beta)}\lbrack S_k\rbrack $ from the subgraph induced on $ P_i $ and $ Q_j $\;
                  \If{$ G_{(\alpha,\beta)}\lbrack S_k\rbrack $ exists}{
                      $ ans.add(G_{(\alpha,\beta)}\lbrack S_k\rbrack)$\;
                      $ max\leftarrow(\vert S_k\vert) $\;
                      }
             }
      }
return ans
\end{algorithm}

Based on the above discussions, we design $Dec$ as shown in Algorithm 4. We first generate candidate keyword sets $ \psi $ and $ \varphi $ respectively using $ S $ and $ q^{'} $ neighbors, $ P_i $ denote the set of nodes containing $ \psi_i $ and $ Q_j $ denote the set of nodes containing $ \varphi_j $ (line 1-2). Next, we update $ \psi, \varphi, P, Q $ through removing the vertex sets and the corresponding keyword sets that dissatisfy structure cohesiveness constraint (line 3). Then, we set $ max = 0 $, indicating the maximal size of all candidate keyword sets, and initialize  set $ S $ and $ c $, where $ S $ contains $ c $ and $ c $ denotes a set consisting of a keyword set, $\psi^{'}$,from $ \psi $ and a keyword set, $\varphi^{'}$, from $\varphi $ (line 4). $ \forall\psi^{'}\in\varphi $ and $ \forall\varphi^{'}\in\varphi $, we generate ($ \vert\psi\vert\times\vert\varphi\vert $) $c$ and put them into $ S $ (line 5-8). For each subset of $S$, we sort it in descending order according to the number of elements in it (line 9). After that, while $ S_k\in S $ and $ \vert S_k\vert>max $, we verify the existence of $ G_{(\alpha,\beta)}\lbrack S_k\rbrack $ in order. If $ G_{(\alpha,\beta)}\lbrack S_k\rbrack $ exists, we put it into the set $ans$ and replace $max$ by $\vert S_k\vert$.For the rest set in $ S $, when we find a set with less than $max$ elements, we stop the verification and output the desired $(\alpha,\beta)-$communities in $ans$.

\begin{theorem}\label{thm4}
Given a bipartite graph G, $Dec$ computes $ G_{\alpha,\beta}\lbrack S_k\rbrack $ in $O((2^{\vert S\vert}\cdot d\cdot2^{\vert S_v\vert_{max}})\cdot O(compute\ G_{\alpha,\beta}(q,G))+n_v\cdot d\cdot2^{\vert S_v\vert_{max}}\log(d\cdot2^{\vert S_v\vert_{max}}))$.
\end{theorem}

\begin{proof}
In Algorithm 4, we use $d$ to represent the degree of $q$, $\vert S_v\vert_{max}$ to represent the $ W_V(v) $ of largest size among all $v\in N(q)$ , we can initialize $ \psi $ and $ \varphi $ in $O(2^{\vert S\vert}+d\cdot2^{\vert S_v\vert_{max}}\log(d\cdot2^{\vert S_v\vert_{max}}))$ time. Line 2-3 can be completed in $O(n_u\cdot2^{\vert S\vert}+n_v\cdot d\cdot2^{\vert S_v\vert_{max}}\log(d\cdot2^{\vert S_v\vert_{max}}))$ time. In line 5-8, set $c$ can be generated in $O(2^{\vert S\vert}\cdot d\cdot2^{\vert S_v\vert_{max}}))$ time. Then it takes $O(2^{\vert S\vert}\cdot d\cdot2^{\vert S_v\vert_{max}}\log(2^{\vert S\vert}\cdot d\cdot2^{\vert S_v\vert_{max}}))$ sorting $S$ in descending order of the number of elements in $S$. In the worst case, it costs $O((2^{\vert S\vert}\cdot d\cdot2^{\vert S_v\vert_{max}})\cdot O(compute\ G_{\alpha,\beta}(q,G)))$ to find all qualified $ G_{(\alpha,\beta)}\lbrack S_k\rbrack $ in line 10-18. However, it will be much faster in practice.
\end{proof}
\section{Experiments}
This section presents our experimental results. We evaluate the efficiency of the techniques for retrieving attributed $(\alpha,\beta)$-communities.

\subsection{Experiments setting}
\noindent\textbf{Algorithms.} We implement and compare following algorithms: 1) a baseline algorithm $Basic$ we propose in Section 4, 2) an improved algorithm $Basic^{+}$ based on Basic,3) the improved attributed $(\alpha,\beta)$-community search algorithm $Inc$, 4) the improved attributed $(\alpha,\beta)$-community search algorithm  $Dec$ in Section 5.

\noindent\textbf{Datasets.}
We evaluate the algorithms on eight real graphs which are $ Enwikibooks$ , $Movie$, $IMDB$, $Actor$, $Discogs$, $Idwiki$, $Plwiki$ and $Nlwiki$. All the datasets we use can be found in  KONECT (http://konect.cc/networks). Note that, for the datasets without attributes, we respectively generate two different kinds of keyword sets for the vertices in the different layer of the bipartite graphs. In each experiment we randomly select 8-13 keywords (average 10) for each vertex. The summary of datasets is shown in Table 1. $U$ and $V$ are vertex layers, $\vert E\vert$ is the number of edges, and $ \widehat d $ is the average degree of vertices.

\begin{table}[th]
	\caption{Datasets used in our experiments}
		\label{tab1}
	\begin{center}
	 \resizebox{85mm}{16mm}{
		\begin{tabular}{|c|c|c|c|c|c|}
			\toprule[1.3pt]
			\textbf{ID} & \textbf{Dataset}                 & \textbf{$\left|U\right|$} & \textbf{$\left|V\right|$} & \textbf{$\left|E\right|$} & \textbf{$ \widehat d $} \\ \hline
D0          & Enwikibooks(Wikibooks edits) & 79,268  & 249,725                   & 766,272                   & 4.66                    \\ 
D1          & Movie(Actor movies)              & 127,823 & 383,640                   & 1,470,404                 & 5.75                    \\ 
D2          & IMDB(komarix-imdb)               & 685,568 & 186,414                   & 2,715,604                 & 6.23                    \\ 
D3          & Actor(actor2)                    & 303,617 & 896,302                 & 3,782,463                 & 6.30                    \\ 
D4          & Discogs(Discogs)                 & 1,754,823 & 270,771                 & 5,302,276                 & 5.24                    \\ 
D5          & Idwiki(edit-idwiki)              & 125,481 & 2,183,494                & 6,126,592                 & 5.31                    \\ 
D6          & Plwiki(edit-plwiki)              & 207,781 & 2,664,432                 & 21,219,204                & 14.78                   \\ 
D7          & Nlwiki(edit-nlwiki)              & 220,847 & 3,800,349                 & 22,142,951                & 11.01                   \\ 
			\bottomrule[1.3pt]
         \end{tabular}}
	\end{center}
\end{table}

The algorithms are implemented in C++ and the experiments are run on a machine having two tetradeca-core Intel Xeon E5-2680 v4 processor, and 251GB of memory, with Ubuntu installed. We set the maximum running time for each test to be 3 days. If a test does not stop in the time limit, we denote the corresponding processing time as INF. The code is open-sourced in https://github.com/892681347/AttributeBigraph.

\subsection{ Evaluation of retrieving attributed $(\alpha,\beta)$-community}
Here we evaluate the performance of the algorithms ($Basic$, $Basic^{+}$, $Inc$ and $Dec$) for querying attributed $(\alpha,\beta)$-communities. We set the default values of $\alpha$ and $\beta$ to 3, and the input keyword set S is set to be the full set of keywords contained in the query vertex. For each dataset, we randomly select 300 query vertices with core numbers greater than or equal to the core number we set. The value of each data is the average result of those 300 queries. For each dataset, we also randomly select $ 20\% $, $ 40\% $, $ 60\% $ and $ 80\% $ of its vertices and obtain four subgraphs induced by these vertex sets, $ 20\% $, $ 40\% $, $ 60\% $ and $ 80\% $ of its keywords and obtain four keyword sets.

The running time of $Basic$ is more than 3 days for all experiments, while the $Basic^{+}$ is unpredictable for large graphs (Idwiki, Plwiki and Nlwiki), so we record them as INF, and the effect of $Basic$ and $Basic^{+}$ algorithm will not be described separately in the corresponding experiments.

\noindent\textbf{Evaluating the effect of query parameters $ \alpha $ and $ \beta $.} We vary $\alpha$ and $\beta$ to assess the performance of these algorithms. In Fig.7(a)-7(h), $\beta$ is  fixed and the experimental parameter $\alpha$ gradually increases from 2 to 6.  We can observe that as $\alpha$ keeps increasing, the running time of $Basic^{+}$, $Inc$ and $Dec$ algorithms decreases. This is because only a few number of vertices and edges are removed from the original graph when the query parameter $\alpha$ is small. When $\alpha$ is large, the resulting $(\alpha,\beta)$-communities are much smaller than the original graph. Thus the size of subgraph directly impacts on the running time of $Basic^{+}$, $Inc$ and $Dec$ algorithms. Obviously, $Dec$ algorithm takes less time than $Basic^{+}$ and $Inc$ algorithms in any case. In Fig.8(a)-8(h), we fix $\alpha$ and vary $\beta$ to compare the query efficiency. In the experiment, we gradually increase the experimental parameter $\beta$ from 2 to 6 and the experimental results are similar to those when $\alpha$ increases. With the increase of $\beta$, the running time of $Basic^{+}$, $Inc$ and $Dec$ algorithms decreases. This is also because higher $\beta$ returns a subgraph with less vertices from the original graph, while $Basic^{+}$ and $Inc$ algorithms are easier to be affected by the number of vertices.

\begin{figure}[H]\label{fig}
\centering
\includegraphics[width=85mm]{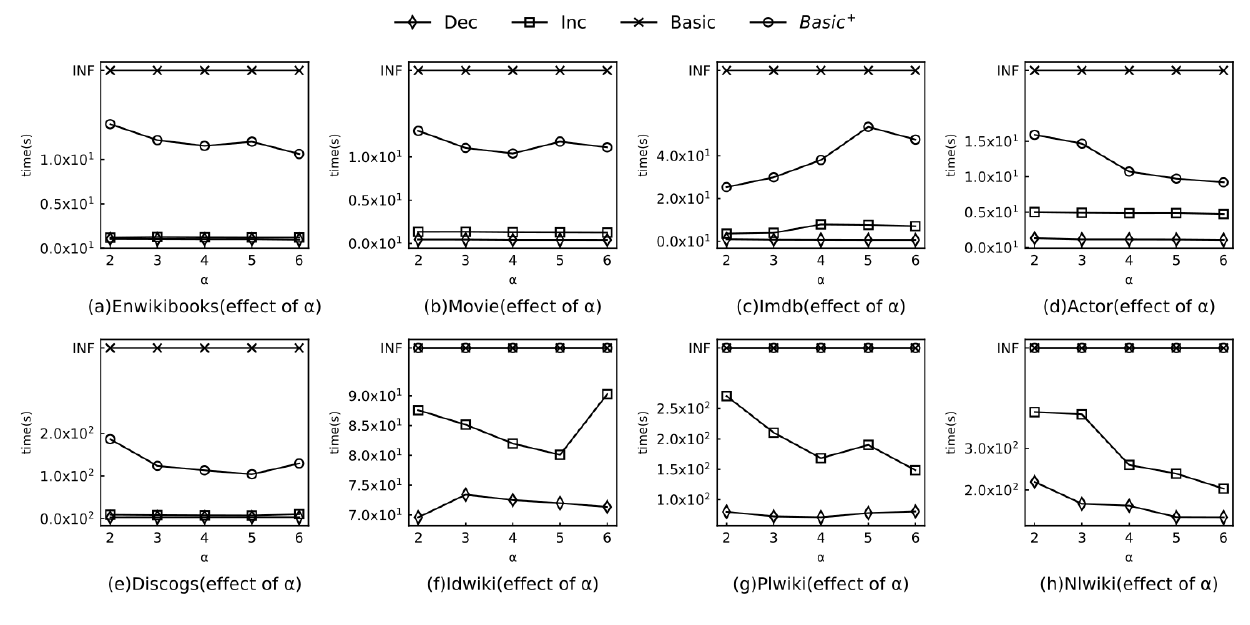}
\vspace{-0.6cm}
\caption{Effect of $ \alpha $}
\vspace{-0.2cm}
\end{figure}

\begin{figure}\label{fig}
\centering
\includegraphics[width=85mm]{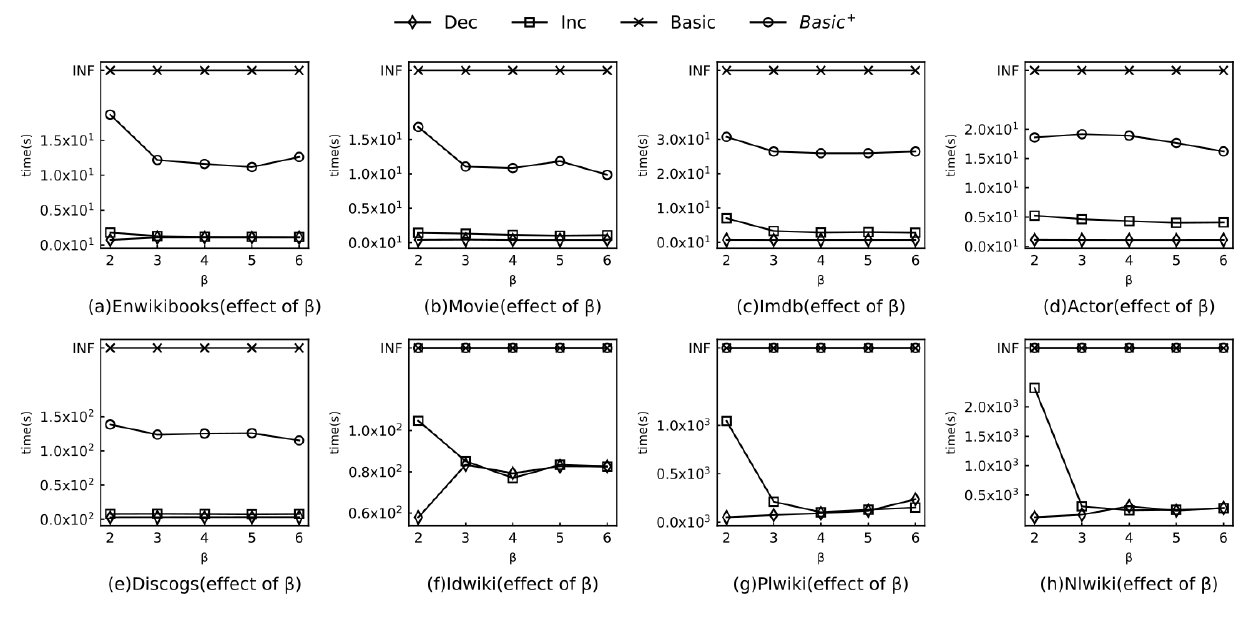}
\vspace{-0.6cm}
\caption{Effect of $ \beta $}
\vspace{-0.2cm}
\end{figure}

\noindent\textbf{Evaluating the scalability w.r.t. keyword.} In this experiment, we evaluate scalability over the fraction of keywords for each vertex. We vary the number of keywords by randomly sampling them from $20\%$ to $100\%$. As shown in Fig.9(a)-9(h), when varying the number of keywords, the running time of $Basic^{+}$, $Inc$ and $Dec$ algorithms stably increases. This is because when the number of keywords increase, the number of subgraphs derived from the keywords and the vertices and edges in each subgraph will increase accordingly. The running time of $Basic^{+}$ and $Inc$ algorithms increase faster than that of $Dec$ algorithm as more keywords are involved, which indicates that $Dec$ performs the better and has a good scalability in practice.

\begin{figure}\label{fig}
\centering
\includegraphics[width=85mm]{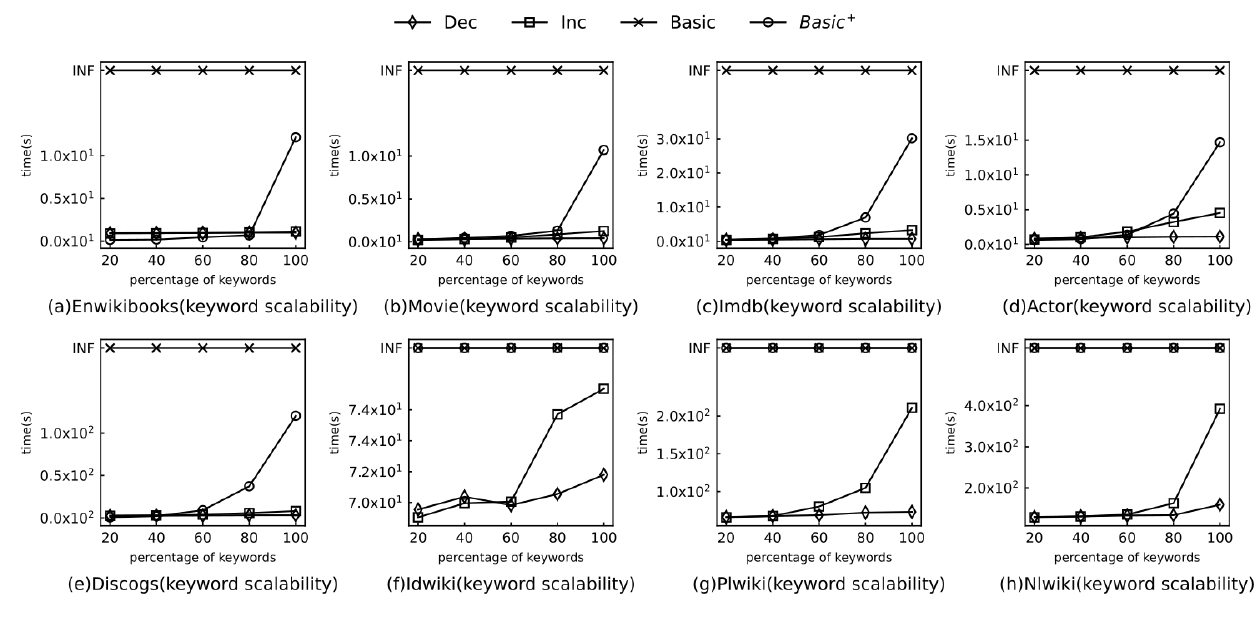}
\vspace{-0.6cm}
\caption{Scalability w.r.t. keyword}
\vspace{-0.2cm}
\end{figure}

\begin{figure}\label{fig}
\centering
\includegraphics[width=85mm]{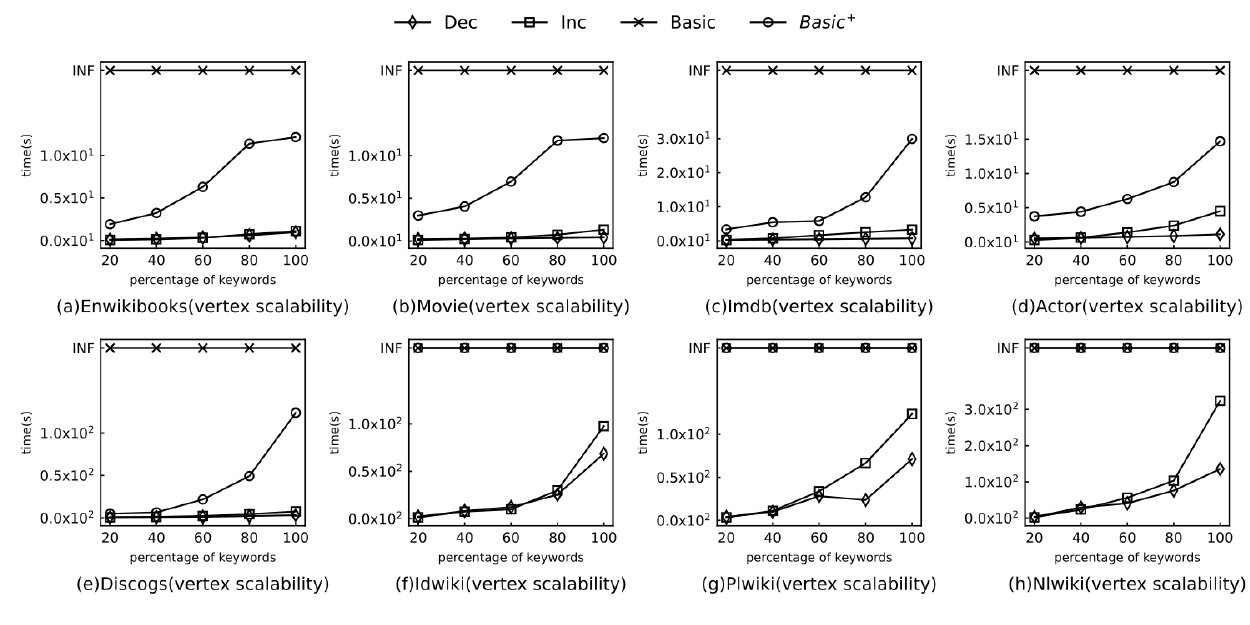}
\vspace{-0.6cm}
\caption{Scalability w.r.t. vertex}
\vspace{-0.2cm}
\end{figure}

\begin{figure}\label{fig}
\centering
\includegraphics[width=85mm]{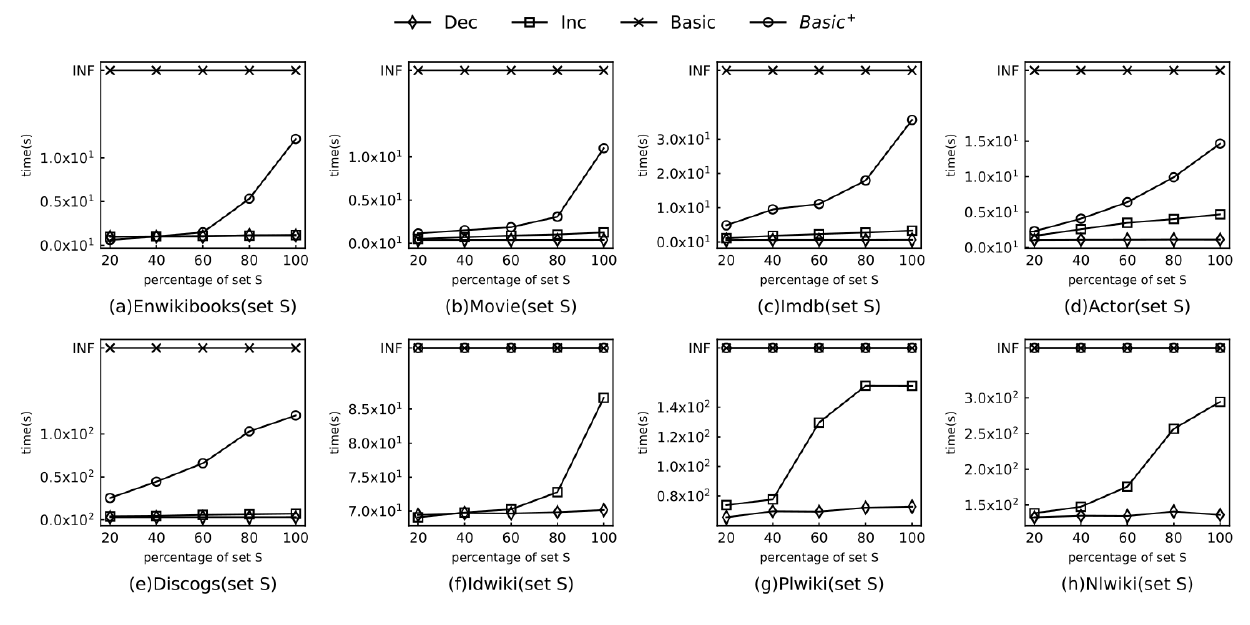}
\vspace{-0.6cm}
\caption{Effect of set S}
\vspace{-0.2cm}
\end{figure}

\noindent\textbf{Evaluating the scalability w.r.t. vertex.} In this experiment, we evaluate the scalability over different fraction of vertices. To test the scalability, we vary the number of vertices and edges by randomly sampling them respectively from $20\%$ to $100\%$ and keeping the induced subgraphs as the input graphs. All the keywords of vertices are considered. Fig.10(a)-10(h) show that, as the number of vertices increasing from $ 20\%$ to $100\%$, the running time for $Basic^{+}$, $Inc$ and $Dec$ algorithms stably increases, and the running time of $Basic^{+}$ and $Inc$ increases faster than that of $Dec$. For example, on Imdb, When the number of nodes increases from $ 20\%$ to $100\%$, the running time of $Dec$ increase from 0.30s to 0.75s, while that of $Basic^{+}$ increase from 3.38s to 29.93s and that of $Inc$ increase from 0.28s to 3.32s. We see that $Dec$ has better performance than $Inc$ for most cases, but the opposite may occur in some cases with few vertices. This is because $Inc$ algorithm is easier to be affected by the number of vertices than $Dec$.

\noindent\textbf{Evaluating the effect of $ S $.} In this experiment, we evaluate the effect of the experimental parameter $S$ on the efficiency of the algorithms. For each query vertex, we randomly sampling $20\%$, $40\%$, $60\%$, $80\%$ and $100\%$ keywords of it to form the query keyword set $S$. As shown in Fig.11(a)-11(h), We can see that with the increase of $\vert S\vert$, the running time of $Basic^{+}$ and $Inc$ increase rapidly, while that of $Dec$ algorithm increases slowly or almost unchanges. For example, on Actor, the running time of $Dec$ increase form 1.08s to 1.13s, while that of $Basic^{+}$ increase form 2.32s to 14.68s and that of $Inc$ increase form 1.65s to 4.70s. The result shows that $Dec$ performs better than $Basic$ and $Inc$.

\noindent\textbf{Case study.} We conduct queries on the real dateset Southern women (small) from the KONECT (http://konect.cc/networks/), where each vertex in $U$ represents a woman, each vertex in $V$ represents a social activity and each edge indicates the woman participates in the social activity.

\begin{figure}[H]\label{fig}
\centering
\includegraphics[width=80mm]{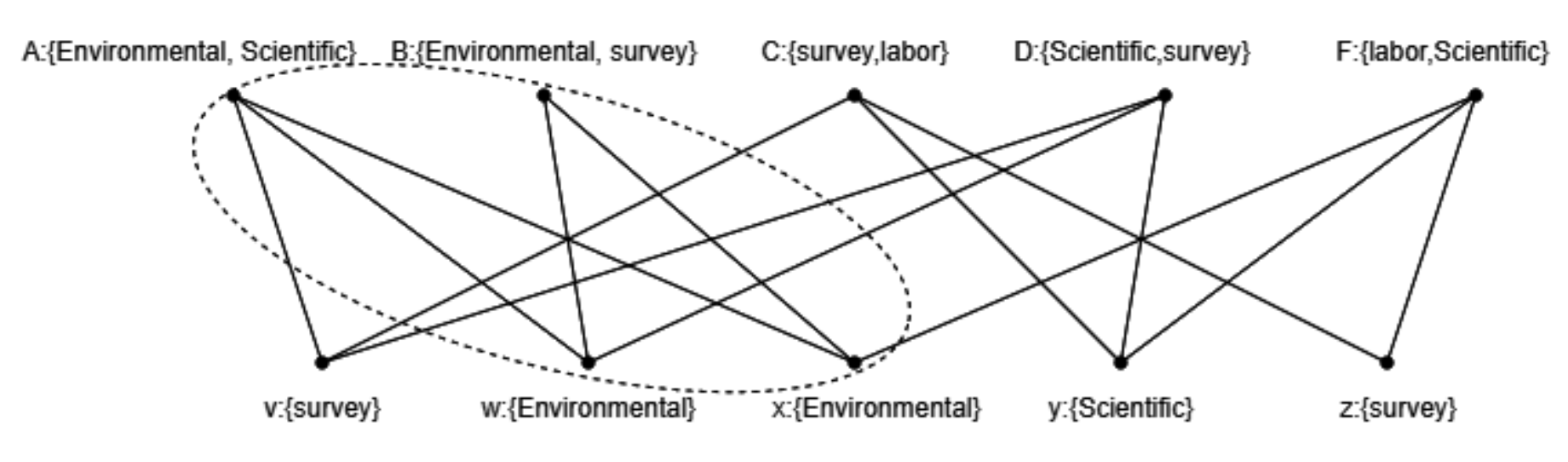}
\vspace{-0.2cm}
\caption{A real person-activity network}
\vspace{-0.5cm}
\end{figure}

We use $A$ as a query vertex, $\alpha$ and $\beta$ are both set to 2, and $S$ contains the keyword “environmental”, the query result is shown in the circled part containing women $\{A, B\}$ and activities $\{w, x\}$ as Fig.11 shows. From the result, we can see the returned people $A$ and $B$ are active participants in environmental activities, and the social activities $w$ and $x$ are all environmental activities with multiple participants from U. In this case, if there is an environmental social activity that needs to recruit team members, then $A$ and $B$ can be given priority because they not only have a preference for environmental social activities but also have experience of cooperation among team members. If we search an (2,2)-community without considering keywords, the result will return the whole women and activities in Fig.12, which includes those who do not often participate in environmental activities. Obviously, the returned candidates cannot be valid team members expected by an environmental activity. This is because we only consider the structure cohesiveness constraint but ignore the keyword cohesiveness constraint.

\section{Conclusion}
In this paper, we study the attributed $(\alpha,\beta)$-community search problem. To solve this problem efficiently, we follow a two-step framework which first generates candidate keyword sets, and then verifies the existence of attributed $(\alpha,\beta)$-community according to each candidate keyword set. Then we develop a basic and two improved query algorithms to retrieve the $(\alpha,\beta)$-community through verifying the candidate keyword sets in a different order.We conduct extensive experiments on real-world graphs, and the results demonstrate the effectiveness of the attributed $(\alpha,\beta)$-community model and the proposed techniques.

\bibliographystyle{ACM-Reference-Format}
\bibliography{ref}

\end{document}